%% file: CSV20.tex
\documentclass[sigconf]{acmart}
\pdfoutput=1
\usepackage{amssymb}
\usepackage{amsmath}
\usepackage{amsthm}
\usepackage{epsfig}
\usepackage{graphicx} 
\usepackage{graphics, color}
\usepackage{hyperref}
\usepackage{algorithm}
\usepackage{algpseudocode}
\usepackage{url}
\newfloat{algorithm}{t}{lop}
\usepackage{etoolbox}
\usepackage{verbatim}
\usepackage{listings}
\usepackage{tabularx}
\usepackage{multirow}
\usepackage{subcaption}
\usepackage{bbm}
\usepackage{times}
\usepackage{xcolor}
\usepackage{soul}
\usepackage[utf8]{inputenc}
\usepackage{thm-restate}

\newcommand{\ts}{\ensuremath{\mathsf{TraceSampler}}}
\newcommand{\ds}{\ensuremath{\mathsf{drawSample}}}
\newcommand{\dsr}{\ensuremath{\mathsf{drawSample\_rec}}}
\newcommand{\sfa}{\ensuremath{\mathsf{sampleFromADD}}}
\newcommand{\uni}{\ensuremath{\mathsf{UniGen}}}
\newcommand{\utwo}{\ensuremath{\mathsf{UniGen2}}}
\newcommand{\uniw}{\ensuremath{\mathsf{UniWit}}}
\newcommand{\amc}{\ensuremath{\mathsf{ApproxMC}}}
\newcommand{\athree}{\ensuremath{\mathsf{ApproxMC3}}}
\newcommand{\tr}{\ensuremath{\omega}}
\newcommand{\Tr}{\ensuremath{\Omega}}
\newcommand{\Path}{\ensuremath{\pi}}
\newcommand{\Paths}{\ensuremath{\Pi}}
\newcommand{\lo}{\ensuremath{lo}}
\newcommand{\hi}{\ensuremath{hi}}
\newcommand{\mi}{\ensuremath{mid}}
\newcommand{\lvlDiff}{\ensuremath{\gamma}}
\newcommand{\pc}[3]{\ensuremath{c(#1,#2,#3)}}

\newcommand{\rso}{\ensuremath{\alpha}}
\newcommand{\rst}{\ensuremath{\beta}}

\newcommand{\version}[2]{#1}

\algnewcommand\algorithmicinput{\textbf{Input:}}
\algnewcommand\INPUT{\item[\algorithmicinput]}
\algnewcommand\algorithmicoutput{\textbf{Output:}}
\algnewcommand\OUTPUT{\item[\algorithmicoutput]}

\algnotext{EndFor}
\algnotext{EndIf}
\algnotext{EndWhile}
\algnotext{Until}
\title{On Uniformly Sampling Traces of a Transition System (Extended Version)}
\thanks{Author names are sorted alphabetically and do not reflect extent of contribution}
\thanks{This is the author's version of the work which includes additional information and analyses not present in the peer-reviewed version. It is posted here for	your personal use. Not for redistribution. The definitive version is to be published in proceedings of {ICCAD 2020}, https://doi.org/10.1145/3400302.3415707.}
\author{Supratik Chakraborty}
\affiliation{%
	\department{Dept. of Computer Science \& Engg.}
	\institution{IIT Bombay}
	\city{Mumbai}
	\country{India}
}
\email{supratik@cse.iitb.ac.in}
\author{Aditya A. Shrotri}
\affiliation{%
	\department{Department of Computer Science}
	\institution{Rice University}
	\city{Houston}
	\country{USA}
}
\email{as128@rice.edu}
\author{Moshe Y. Vardi}
\affiliation{%
	\department{Department of Computer Science}
	\institution{Rice University}
	\city{Houston}
	\country{USA}
}
\email{vardi@rice.edu}

\copyrightyear{2020}
\acmYear{2020}
\setcopyright{acmlicensed}\acmConference[ICCAD '20]{IEEE/ACM International Conference on Computer-Aided Design}{November 2--5, 2020}{Virtual Event, USA}
\acmBooktitle{IEEE/ACM International Conference on Computer-Aided Design (ICCAD '20), November 2--5, 2020, Virtual Event, USA}
\acmPrice{15.00}
\acmDOI{10.1145/3400302.3415707}
\acmISBN{978-1-4503-8026-3/20/11}
\begin{document}
	\input{Abstract}
	\maketitle
	
	\input{Introduction}

	\input{Notation}
	\input{Related}
	\input{Algorithm}
	\input{IterSquare}
	\input{Analysis}

	\input{Experiments}
	\input{Conclusion}
	\begin{acks}
	The authors would like to thank Krishna S. and Shaan Vaidya for discussions on some initial ideas. This work was supported in part by an MHRD IMPRINT-1 grant (Project 6537) of Govt of India and NSF grants IIS-1527668, CCF-1704883, IIS-1830549, and an award from the Maryland Procurement Office.
	\end{acks}
	\bibliographystyle{ACM-Reference-Format}
	\bibliography{Refs}
\end{document}

%% file: Abstract.tex
\begin{abstract}
  A key problem in constrained random verification (CRV) concerns generation
  of input stimuli 
  that result in
  good coverage 
  of the system's runs in targeted corners of its
  behavior space.  Existing CRV solutions 
  however provide no formal guarantees on the distribution of the
  system's runs.
  In this paper, we take a first step towards solving this problem.
  We present an algorithm based on Algebraic Decision Diagrams for
  sampling bounded \emph{traces} (i.e. sequences of states) of a
  sequential circuit with provable uniformity (or bias) guarantees,
  while satisfying given constraints.  
  We have implemented our algorithm in a tool called {\ts}.
  Extensive experiments show that {\ts} outperforms alternative approaches
  that provide similar uniformity guarantees. 
\end{abstract}

%% file: Introduction.tex
\section{Introduction}
\label{sec:intro}
\emph{Simulation-based functional verification} is a crucial yet
time-consuming step in modern electronic design automation
flows~\cite{foster2015trends}.  In this step, a design is simulated
with a large number of input stimuli, and signals are monitored to
determine if coverage goals and/or functional requirements are met.
For complex designs, each input stimulus typically spans a large
number of clock cycles.  Since exhaustive simulation is impractical
for real designs, using ``good quality'' stimuli that result in
adequate coverage of the system's runs in targeted corners is
extremely important~\cite{bening2001principles}.  \emph{Constrained
  random verification, or
  CRV,}~\cite{yuan2006,Bhadra2007,kitchen2007,naveh2007constraint}
offers a practical solution to this problem.  In CRV, the user
provides constraints to ensure that the generated stimuli are valid
and also to steer the system towards bug-prone corners.  To ensure
\emph{diversity}, 
CRV allows randomization in the choice of stimuli satisfying a set of
constraints.  This can be very useful when the exact inputs needed to
meet coverage goals or to test functional requirements are not
known~\cite{kitchen2007,recent-diversity-paper}.  In such cases, it is
best to generate stimuli such that the resulting runs are uniformly
distributed in the targeted corners of its behavior space.
Unfortunately, state-of-the-art CRV
tools~\cite{uvm,e,survey,systemverilog1,systemverilog2} do not permit
such uniform random sampling of input stimuli.  Instead, they allow
inputs to be assigned random values from a constrained set at specific
simulation steps.  This of course lends diversity to the generated
stimuli.  However, it gives no guarantees on the distribution of the
resulting system runs. In this paper, we take a first step towards
remedying this problem.  Specifically, we present a technique for
generating input stimuli that \emph{guarantees} uniform (or
user-specified bias in) distribution of the resulting system runs. 
\version{Note that this is significantly harder than generating any one run satisfying a set of constraints.}{}

We represent a run of the system by the sequence of states through
which it transitions in response to a (multi-cycle) input stimulus.
Important coverage metrics (viz. transition coverage, state sequence
coverage, etc.~\cite{cov-cite}) are usually boosted by choosing
stimuli that run the system through diverse state sequences.
Similarly, functional requirements (viz. assertions in
SystemVerilog~\cite{systemverilog2}, PSL~\cite{ieee1850_2010}, Specman
E~\cite{e}, UVM~\cite{uvm} and other formalisms~\cite{fr-cite}) are
often stated in terms of temporal relations between states in a run of
the system.  Enhancing the diversity of state sequences in runs 
therefore improves the chances of detecting violations, if any,
of functional requirements.  Consequently, generating input stimuli
such that the resulting sequences of states, or \emph{traces}, are
uniformly distributed among all traces consistent with the given
constraints is an important problem.  Significantly, given a sequence
of states and the next-state transition function, the input stimuli
needed to induce the required state transitions at each clock cycle
can be easily obtained by independent SAT/SMT calls for each cycle.
Hence, our focus in the remainder of the paper is the core problem of
sampling a system's traces uniformly at random from the set of all
traces (of a given length) that satisfy user-specified constraints.

To see why state-of-the-art CRV
techniques~\cite{uvm,e,survey,systemverilog1,systemverilog2} often
fail to generate stimuli that produce a uniform distribution of
traces, consider the sequential circuit with two 
latches ($x_0$ and $x_1$) and one primary input, shown in
Fig.~\ref{fig:ckt}a.  The state transition diagram of the circuit is
shown in Fig.~\ref{fig:ckt}b. Suppose we wish to uniformly sample
traces that start from the initial state $s_0 = (x_1=0, x_0=0)$ and
have $4$ consecutive state transitions.  From Fig.~\ref{fig:ckt}b,
there are $7$ such traces: $ \tr_1 = s_0s_1s_1s_1s_1$, $ \tr_2 =
s_0s_1s_1s_1s_2$, $ \tr_3 = s_0s_1s_1s_2s_2$, $\tr_4 =
s_0s_1s_2s_2s_2$, $\tr_5 = s_0s_3s_1s_1s_1$, $\tr_6 = s_0s_3s_1s_1s_2$
and $\tr_7 = s_0s_3s_1s_2s_2$.  Hence, each of these traces must be
sampled with probability $1/7$. Unfortunately, the state transition
diagram of a sequential circuit can be exponentially large (in number
of latches), and is often infeasible to construct explicitly.  Hence
we must sample traces without generating the state transition diagram
explicitly.  The primary facility in existing CRV techniques to
attempt such sampling is to choose values of designated inputs
randomly at specific steps of the simulation.  In our example, without
any information about the state transition diagram, the primary input
of the circuit in Fig.~\ref{fig:ckt}a must be assigned a value 0 (or
1) with probability $1/2$ independently in each of the $4$ steps of
simulation.  This produces the traces $ \tr_1$ and $\tr_2$ with
probability $1/16$ each, $ \tr_3$, $\tr_5$ and $\tr_6$ with
probability $1/8$ each, and $\tr_4$ and $\tr_7$ with probability $1/4$
each.  Notice that this is far from the desired uniform distribution.
In fact, it can be shown that for every
choice of bias for sampling $0/1$ values of the primary input at each
state, we get a non-uniform distribution of $\tr_1$ through $\tr_7$.
\begin{figure}
	\centering
	\includegraphics[width=2.8in]{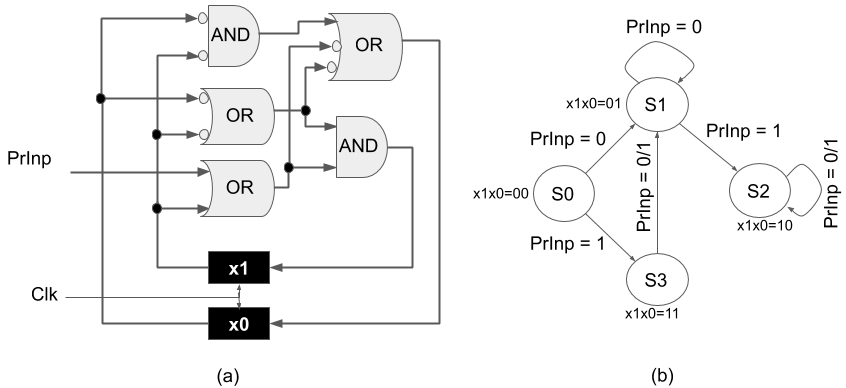}
	\caption{(a) Sequential circuit, (b) State transition diagram}
	\label{fig:ckt}
	\vspace{-0.5cm}
\end{figure}

The trace-sampling problem can be shown to be at least as hard as
uniformly sampling satisfying assignments of Boolean formulas.  The
complexity of the latter problem has been extensively
studied~\cite{Sipser83,jerrum1986random,BGP2000}, and no efficient
algorithms are known.  Therefore, efficient algorithms for sampling
traces are unlikely to exist.  Nevertheless, a trace sampling
technique that works efficiently in practice for many problem
instances is likely to be useful even beyond CRV, viz. in test
generation using Bounded Model Checking~\cite{hamon2004generating}.

The primary contributions of this paper are as follows:
\begin{enumerate}
  \item A novel algorithm for sampling fixed-length traces of a
    transition system using Algebraic Decision Diagrams
    (ADDs)~\cite{bahar1997algebric}, with provable guarantees of
    uniformity (or user-provided bias). The following are distinctive
    features of our algorithm.
    \begin{enumerate}
      \item It uses iterative squaring, thereby requiring only $\log_2
        N$ ADDs to be pre-computed when sampling traces of $N$
        consecutive state transitions.  This 
        allows our algorithm to scale to traces of a few hundred
        transitions in our experiments.
      \item It is easily adapted when the trace length is not
        a power of $2$, and also when implementing weighted
        sampling of traces with multiplicative weights.
      \item It pre-compiles the $i$-step transition relation for
        $\log_2 N$ different values of $i$ to ADDs.  This allows it to
        quickly generate multiple trace samples once the ADDs are
        constructed. Thus the cost of ADD construction gets amortized
        over the number of samples, which is beneficial in CRV
        settings.
    \end{enumerate}
    \item A comparative study of an implementation of our algorithm
      (called {\ts}) with alternative approaches based on
      (almost)-uniform sampling of propositional models, that provide
      similar uniformity guarantees.  Our experiments demonstrate that
      our approach offers significant speedup and is the fastest over
      90\% of the benchmarks.
\end{enumerate}


%% file: Notation.tex
\section{Preliminaries}
\label{sec:prelim}
\subsection{Transition Systems and Traces}
A synchronous sequential circuit with $n$ latches implicitly
represents a transition system with $2^n$ states.  Hence, synchronous
sequential circuits serve as succinct representations of finite-state
transition systems. We use ``sequential circuits'' and ``transition
systems'' interchangeably in the remainder of the paper to refer to
such systems.

Formally, a transition system with $k$ Boolean state variables $X =
\{x_0,\ldots,x_{k-1} \}$ and $m$ primary inputs is a $5$-tuple $
(S,\Sigma,t,I,F)$, where $S = \{0,1\}^k $ is the set of states,
$\Sigma = \{0,1\}^m$ is the input alphabet, $I \subseteq S$ is the set
of initial states, $F \subseteq S$ is the set of target (or final)
states, and $t:S\times\Sigma\rightarrow S $ is the state transition
function such that $t(s,a) = s' $ iff there is a transition from state
$ s\in S $ on input $ a \in \Sigma $ to state $s' \in S$.  We view
each state in $S = \{0,1\}^k$ as a valuation of $(x_{k-1} \ldots
x_0)$.  For notational convenience, we use the decimal representation
of the valuation of $(x_{k-1} \ldots x_0)$ as a subscript to refer to
individual states. For instance, $s_0$ and $s_{2^k-1}$ are the states
with all-zero and all-one assignments to $x_{k-1} \ldots x_0$
respectively. We refer to multiple versions of the state
variables $X$ as $ X^0, X^1,\ldots $

Given a transition system, a \emph{trace} $\tr$ of length $ N~(> 0) $
is a sequence of $N+1$ states such that $ \tr[0] \in I $, $ \tr[N]\in
F $ and $ \forall i\in\{0,\ldots,N-1\}\,\,\, \exists a \in \Sigma
\,\,\,s.t.\,\,\, t(\tr[i],a) = \tr[i+1] $, where $ \tr[i] $ represents
the $i^{th}$ state in the trace. We denote the set of all traces of
length $N$ by $ \Tr_N $.  Given a trace $\tr \in \Tr_N$, finding an
input sequence $\alpha \in \Sigma^N$ such that the $i^{th}$ element,
viz. $\alpha[i]$, satisfies $\tr[i+1] = t(\tr[i], \alpha[i])$ for all
$i \in \{0, \ldots N-1\}$, requires $N$ independent SAT solver
calls. 
With state-of-the-art SAT solvers~\cite{soos2009}, this is unlikely to be a
concern with the number of primary inputs $m$ ranging upto tens of
thousands.  Therefore, finding a sequence of inputs that induces a
trace is relatively straightforward, and we will not dwell on this any
further.  Our goal, instead, will be to sample a trace $ \tr \in \Tr_N
$ uniformly at random. Formally, if the random variable $Y$
corresponds to a random choice of traces, we'd like to have
$	\forall \tr\in \Tr_N \,\,\,\, \Pr[Y=\tr] = \frac{1}{|\Tr_N|}$.
Given a weight function $ w:\Tr_N \rightarrow \mathbb{R}^+ $, the
related problem of weighted trace sampling requires us to sample such that
$\forall \tr\in \Tr_N \,\,\,\, \Pr[Y=\tr] =
\frac{w(\tr)}{\sum_{\tr\in\Tr_N} w(\tr)}$.

Since we are concerned only with sequences of states, we will
henceforth assume that transitions of the system are represented by
a transition relation $\widehat{t}: S \times S \rightarrow \{0,1\}$,
where $ \widehat{t}(s,s') \Leftrightarrow \exists a \in
\Sigma\,\,\,s.t.\,\,\,t(s,a) = s'$. For notational convenience, we
abuse notation and use $ t(s,s')$ for $\widehat{t}(s,s')$, when there is
no confusion.

A multiplicative weight function assigns a weight to each state
transition, and defines the weight of a trace as the product of
weights of the transitions in the trace. Formally, let $\widehat{w}: S
\times S \rightarrow \mathbb{R}^{\ge 0} $ be a weight function for
state transitions, where $\widehat{w}(s_i, s_j) > 0$ if $t(s_i, s_j)$
holds, and $\widehat{w}(s_i, s_j) = 0$ otherwise.  Then, the
multiplicative weight of a trace $\tr \in \Tr_N$ is defined as $
w(\tr) = \prod_{i=0}^{N-1} \widehat{w}(\tr[i], \tr[i+1]) $. The
unweighted uniform sampling problem can be seen to be the special case
where $\widehat{w}(s_i, s_j) = 1$ whenever $t(s_i, s_j)$ holds.

\version{
\begin{table}[]

		\begin{tabular}{c|c}
			
			Symbol & Meaning                                    \\ \hline
			$X$        & Set of Boolean variables $x_1, x_2,\ldots,x_k$ \\
			$S$        & Set of states $s_1, s_2, \ldots, s_{2^k-1}$  \\
			$\Omega_N$ & Set of all traces `$ \omega $', of length $N$   \\
			$t$        & Transition function                        \\
			$w$        & Weight function                            \\
			$\Pi_v$    & Set of all paths `$\pi$' in a DD starting at node $ v $               
		\end{tabular}%
		\caption{Summary of notation}
		\label{tab:sym}

\end{table}
}{}
\subsection{Decision Diagrams}
We use Binary Decision Diagrams (BDDs)~\cite{Bryant86} and their
generalizations called Algebraic Decision Diagrams
(ADDs)~\cite{bahar1997algebric} to represent transition
functions/relations and counts of traces of various lengths between
states. 
Formally, both ADDs and BDDs are $4$-tuples $(X,T,\pi,G)$ where $ X $
is a set of Boolean variables, the finite set $ T $ is called the
carrier set, $ \pi: X \rightarrow \mathbb{N} $ is the diagram variable
order, and $ G $ is a rooted directed acyclic graph satisfying the
following properties:
(i) every terminal node of $ G $ is labeled with an element of $ T $, 
(ii) every non-terminal node of $ G $ is labeled with an element of $ X $ and has two outgoing edges labeled $0$ and $1$, and 
(iii) for every path in $ G $, the labels of visited non-terminal nodes must occur in increasing order under $ \pi $.

ADDs and BDDs differ in the carrier set $ T $; for ADDs $ T \subset
\mathbb{R} $ while for BDDs, $ T = \{0,1\} $.
Thus ADDs represent functions of the form $f:
2^X\rightarrow\mathbb{R}$ while BDDs represent functions of the form
$f: 2^X \rightarrow \{0,1\}$, as Directed Acyclic Graphs (DAG). Many
operations on Boolean functions can be performed in polynomial time in
the size of their ADDs/BDDs. This includes conjunction, disjunction,
if-then-else (ITE), existential quantification etc. for BDDs and
product, sum, ITE and additive quantification for ADDs.  The reader is
referred to ~\cite{Bryant86,bahar1997algebric} for more details on
these decision diagrams.

We denote the set of leaves of a decision diagram (DD) $ t $ by $ leaves(t) $, and the
root of the DD by $ root(t) $. We denote the vertices of the DAG by $ v
$, set of parents of $ v $ in the DAG by $ P(v) $, and value of a leaf $ v $ by $ val(v) $. A path
from a node $v$ to $root(t)$ in a DD $ t $, denoted as $ \Path =
v_0v_1\ldots v_h$, is defined to be a sequence of nodes such that
$v_0 = v$, $v_h = root(t) $ and $ \forall i\,\,\, v_{i+1} \in P(v_i) $. We 
use $ \Paths_v $ denote the set of all paths to the root starting from
some node $ v $ in the DD. For a set $V$ of nodes, we define $ \Paths_{V} =
\cup_{v\in V}\,\, \Paths_v $. The special set $ \Paths $ represents 
all paths from all leaves to the root of a DD.
\version{
Our notational setup is briefly summarized in Tab. \ref{tab:sym}.
}{}

%% file: Related.tex
\section{Related Work}
\label{sec:rel}
We did not find any earlier work on sampling traces of sequential
circuits with provable uniformity guarantees. As mentioned earlier,
constrained random verification
tools~\cite{systemverilog1,systemverilog2,e,uvm,survey} permit values
of selected inputs to be chosen uniformly (or with specified bias)
from a constrained set at some steps of simulation.  Nevertheless, as shown
in Section~\ref{sec:intro}, this does not necessarily yield uniform
traces.

Arenas et al.~\cite{arenas2019efficient} gave a fully-polynomial
randomized approximation scheme for approximately counting words of a
given length accepted by a Non-deterministic Finite Automaton
(NFA). Using Jerrum et al's reduction from approximate counting to
sampling~\cite{jerrum1986random}, this yields an algorithm for
sampling words of an NFA. Apart from the obvious difference of
sampling words vs. sampling traces, Arenas et al's technique requires
the state-transition diagram of the NFA to be represented explicitly,
while our focus is on transition systems that implicitly encode
large state-transition diagrams.

Given a transition system, sampling traces of length $N$ can be
achieved by sampling satisfying assignments of the propositional
formula obtained by ``unrolling'' the transition relation $N$ times.
Technique for sampling models of propositional
formulas,
viz. ~\cite{achlioptas2018fast,sharma2018knowledge,gupta2019waps} for
uniform sampling and
~\cite{chakraborty2013scalable,chakraborty2014balancing,chakraborty2015parallel}
for almost uniform sampling, can therefore be used to sample
traces. The primary bottleneck in this approach is the linear growth
of propositional variables with the trace length and count of Boolean
state variables.
We compare our tool with state-of-the-art samplers
WAPS~\cite{gupta2019waps} and {\utwo}~\cite{chakraborty2015parallel},
and show that our approach performs significantly better.

%% file: Algorithm.tex
\section{Algorithms}
\label{sec:alg}

For clarity, we assume that the length of traces, i.e. $N$, is a power of $2$; the case when $N$ is not a power of $2$ is
discussed later.  
A naive approach would be to use a single BDD to represent all traces of length $N$, by appropriately unrolling the transition system, and then sample traces from the BDD. Such monolithic representations, however, are known to blow up~\cite{dudek2019addmc}. Therefore, we use $\log_2 N$ ADDs, where the $ i^{th} $ ADD ($ 1\le i \le \log_2 N $) represents the count of $2^i$-length paths
between different states of the transition system. The $i^{th}$ ADD is constructed from the $(i-1)^{th}$ ADD by a technique similar to iterative
squaring~\cite{burch1990symbolic,burch1990sequential}. A trace is sampled by recursively sampling states from each ADD according to the weights on the leaves. 

The detailed algorithm for constructing ADDs is presented in Algorithm
\ref{alg:madd}. We assume that the transition relation is defined over
$2$ copies, viz. $ X^0 $ and $ X^1 $, of the state variables, and that
an additional $ \log_2 N $ copies, viz. $ X^2\ldots X^{(\log_2 N)+1}
$, are also available. In each step of the for loop on line 2, the
$(i-1)^{th}$ ADD is squared to obtain the $i^{th}$ ADD after
additively abstracting out $X^i$ in line 4. Each ADD $ t_i(X^0,X^i,X^{i+1}) $ represents the count of $2^i$-length traces from $X^0$ to $X^{i+1}$ that pass through $X^i$ at the half-way point. Note that $ g(X^{i-1},X^i)
$ and $ g(X^i,X^{i+1}) $ in line 3 are the same ADD, but with
variables renamed. Finally, in line 5, we take the product of the $
\log_2N^{th}$ ADD with the characteristic functions for the initial and
final states, represented as ADDs. Although Algorithm \ref{alg:madd}
correctly computes all ADDs, in practice, we found that it often
scaled poorly for values of $ N $ beyond a few 10s.  On closer
scrutiny, we found that this was because the ADD $t_0$ (and other ADDs
derived from it) encoded
information about transitions from states unreachable in $ N $ steps
(and hence of no interest to us).  
Therefore, we had to aggressively optimize the ADD computations by
\emph{restricting} (see~\cite{coudert1990verifying}) each ADD $t_i$
with an over-approximation of the set of reachable states relevant to
that $t_i$. We discuss this optimization in detail in Sec. \ref{sec:iter}.

Once the ADDs are constructed, the sampling of the $ N+1 $ states of
the trace is done by Algorithm \ref{alg:ds}. States $\tr[0],\tr[N/2]$
and $\tr[N]$ are sampled from the $ \log_2N^{th}$ ADD in a call to
Algorithm \ref{alg:sfa} in line 2. Then Algorithm \ref{alg:dsr} is
recursively called to sample the first and second halves of the trace
in lines 3 and 4. In each recursive call, Algorithm \ref{alg:dsr}
invokes the procedure in Algorithm \ref{alg:sfa}, to sample the state
at the mid-point of the current segment of the trace under
consideration, and recurses on each of the two halves thus generated.

In {\sfa} (Algorithm \ref{alg:sfa}), the $ \log_2N^{th} $ ADD is used
as-is for sampling (lines 1,2), while other ADDs are first simplified
by substituting the values of state variables in $\tr[lo]$ and
$\tr[hi]$, that have been sampled previously and provided as inputs to
{\sfa} (lines 3,4). The role of the rest of the algorithm is to sample
a path from a leaf to the root in a bottom-up fashion, with
probability proportional to the value of the leaf. Towards this end, a
leaf is first sampled in lines 5-8. We assume access to a procedure $
weighted\_sample $ that takes as input a list of elements and their
corresponding weights, and returns a random element from the list with
probability proportional to its weight. Once a leaf is chosen, we
traverse up the DAG in the loop on line 9. This is done by iteratively
sampling a parent with probability proportional to the number of paths
reaching the parent from the root (lines 10-12). The quantity
$|\Paths_{v}| $ denotes the number of paths from a node $ v $ to the
root, and can be easily computed by dynamic programming. If some
levels are skipped between the current node $ v $ and its parent $ p
$, then the number of paths reaching the current node from the parent
are scaled up by a factor of $ 2^{level(p)-level(v)-1}$ (line 12).
This is because each skipped level contributes a factor of $2$ to the
number of paths reaching the root. Once a parent is sampled, the value
of the corresponding state variable is updated in the trace in lines
13-17, where the procedure $ getTracePosition $ is assumed to return
the index of the state (in the trace $\tr$) and the index of the state
variable (in the set $X$ of state variables) corresponding to the
parent node. $ getTracePosition $ can be implemented by maintaining a
map between the state variables and the variable order in the DD. The
random values for variables in the skipped levels between the parent
and the current node are sampled in lines 18 and 19.

\paragraph*{Non Power-of-2 trace lengths}
When the trace length $N$ is not a power of two, we modify the given
sequential circuit so that the distribution of traces of length $N'~(>
N)$ of the modified circuit is identical to the distribution of
length-$N$ prefixes of these traces.
Conceptually, the modification is depicted in
Fig. \ref{fig:counter}. Here, the ``Saturate-at-N'' counter counts up
from $0$ to $N$ and then stays locked at $N$.  Once the count reaches
$ N $, the next state and current state of the original circuit are
forced to be identical, thanks to the multiplexer.  Therefore, the
modified circuit's trace, when projected on the latches of the
original circuit, behaves exactly like a trace of the original circuit
up to $N$ steps.  Subsequently, the projection remains stuck at the
state reached after $N$ steps. Hence, by using the modified circuit
and by choosing $ N' = 2^{\lceil(\log_2 N)\rceil} $, we can assume
w.l.o.g. that the length of a trace to be sampled is always a power of
2.

\paragraph*{Weighted Sampling}
A salient feature of Algorithms 1-4 is that the same framework can be
used for weighted sampling (instead of uniform) as defined in
Section~\ref{sec:prelim}, with one small modification: if the input $ t_0
$ to Algorithm \ref{alg:madd} is an ADD instead of a BDD, where the
values of leaves are the weights of each transition, then it can be
shown that {\ds} will sample a trace with probability proportional to
its weight, where the weight of a trace is define
multiplicatively as in Section~\ref{sec:prelim}.

\begin{algorithm}
	\caption{$makeADDs(t_0, N, f, I)$}
	\label{alg:madd}
	\begin{algorithmic}[1]
		\INPUT 
		
		$ t_0$: 1-step transition relation
		
		$ N $: trace length 
		
		$ f $: characteristic function of target states
		
		$ I $: characteristic function of initial states
		\OUTPUT ADDs $ t_1 \ldots t_{\log_2 N}$: $ 2^i $-step transition relations $ t_i $
		\State $ g \gets t_0 $;
		\For {$i = 1, 2, \ldots , \log_2 N$}
		\State $ t_i (X^{0},X^i,X^{i+1}) \gets g(X^{0},X^{i}) \times g(X^i,X^{i+1}) $;
		
		\Comment $ \times $ is ADD multiplication
		\State $ g \gets \exists X^i \,\,\, t_i $;
		\Comment Additively abstract vars in $ X^i $
		\EndFor
		\State $ t_{\log_2 N} \gets t_{\log_2 N} \wedge f(X^{(\log_2 N) + 1}) \wedge I(X^0)$
		\State \Return $ t_1 \ldots t_{\log_2N} $
	\end{algorithmic}
\end{algorithm}

\begin{algorithm}
	\caption{$\ds(t_1,\ldots,t_{\log_2 N})$}
	\label{alg:ds}
	\begin{algorithmic}[1]
		\State $ \tr \gets []$ \Comment Initialize empty trace
		
		/* Sample 0, $N/2$ and $N^{th}$ states from $ \log_2 N^{th} $ ADD */
		\State $ \tr[0],\tr[N/2],\tr[N] \gets sampleFromADD(\log_2N,t_{\log_2 N}, 0, N, \tr) $;
		
		/* Sample states $0 \ldots N/2$ */
		\State $\tr[0\ldots N/2]\gets \dsr((\log_2N) - 1, 0, N/2, \tr)$;
		
		/* Sample states $N/2 \ldots N$ */
		\State $\tr[N/2\ldots N]\gets \dsr((\log_2 N) - 1, N/2, N, \tr)$;
		
		\State \Return $ \tr $
	\end{algorithmic}
\end{algorithm}

\begin{algorithm}
	\caption{$\dsr(i, \lo, \hi, \tr, t_1,..,t_{\log_2 N})$}
	\label{alg:dsr}
	\begin{algorithmic}[1]
		\State $ \mi \gets (\lo + \hi)/2 $;
		\State $ \cdot, \tr[\mi], \cdot \gets sampleFromADD(i,t_i, \lo, \hi, \tr) $; 
		
		\Comment Sample $ \tr[\mi] $. ($ \tr[\lo], \tr[\hi] $ unchanged) 
		\State $\tr[\lo\ldots \mi]\gets \dsr(i-1, \lo,\mi,\tr)$;
		\State $\tr[\mi\ldots \hi]\gets \dsr(i - 1, \mi, \hi, \tr)$;
		\State \Return $ \tr $
	\end{algorithmic}
\end{algorithm}

\begin{algorithm}
	\caption{$sampleFromADD(i,t_i, \lo, \hi, \tr)$}
	\label{alg:sfa}
	\begin{algorithmic}[1]
		\State $ \mi \gets (\lo + \hi)/2 $;
		\If {$i==\log_2N$} \Comment Use whole ADD for sampling
			\State $ \hat{t} \gets t_i $;
		\Else \Comment Reduce ADD with states previously sampled
			\State $\hat{t} \gets Substitute(t_i, \tr[\lo], \tr[\hi])$
		\EndIf
		\State $ wtList \gets [] $; \Comment Array for weights
		
		/*Sample a leaf*/
		\For {$ v_l \in leaves(\hat{t}) $ } 
			\State $wtList[l] \gets val(v_l)*|\Paths_{v_l}|$
		\EndFor
		\State $v \gets weighted\_sample(wtList, leaves(\hat{t}))$ 
		
		/*Sample parents up to root*/
		\While {$v \ne root(\hat{t})$}
			\For {$ p \in P(v)$} \Comment Find weights of all parents
				\State $wtList[p] \gets 2^{level(p)-level(v)-1}*|\Paths_{p}|$; 
				
				\Comment Weight adjusted for skipped levels
			\EndFor
			\State $p^* \gets weighted\_sample(wtList, P(v))$;
			\State $ j_s,j_b \gets getTracePosition(p^*, i, \lo, \hi) $;
			
			\Comment $ j_s $ is state index, $ j_b $ is variable bit index
			\If {$then\_child(p^*) == v$} 
				\State $\tr[j_s][j_b] \gets True$
			\Else 
				\State $\tr[j_s][j_b] \gets False$
			\EndIf
			\For {each $v_{skipped}$ between $ p^* $ and $ v $}
				\State $ j_s,j_b \gets getTracePosition(v_{skipped}, i, \lo, \hi) $;
				\State $\tr[j_s][j_b] \gets random\_bit()$ 
				\Comment For skipped vars 
			\EndFor
			\State $v \gets p^*$
		\EndWhile
		\State\Return $ \tr[\lo], \tr[\mi], \tr[\hi] $
	\end{algorithmic}
\end{algorithm}
\begin{figure}
\centering
\includegraphics[width=2.3in]{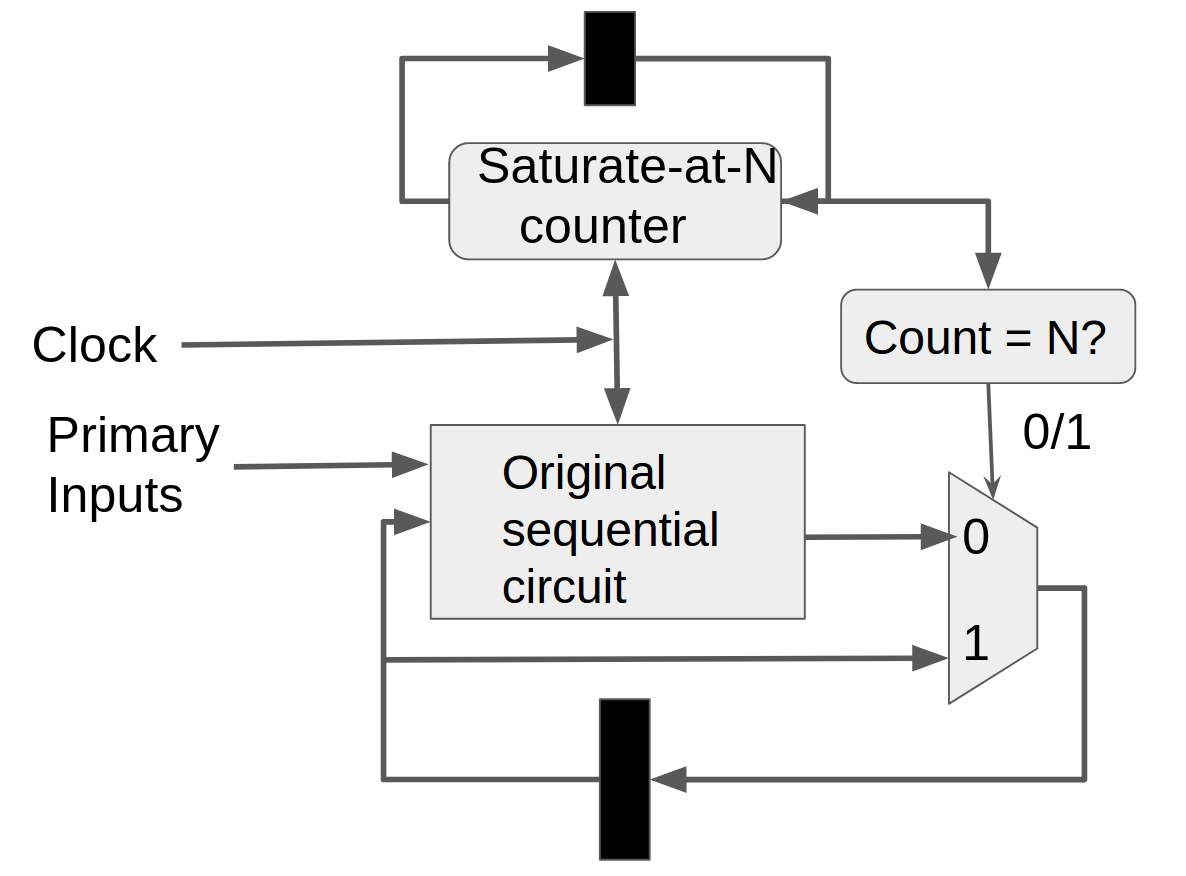}
\caption{Modified Circuit for Non-Power-of-2 Trace Lengths}
\label{fig:counter}
\end{figure}
%
%
%
%
%

%% file: IterSquare.tex
\section{Improved Iterative Squaring}
\label{sec:iter}
In this section, we present a more efficient version of
Alg. \ref{alg:madd}. To see where gains in efficiency can be made,
note that the $ t_{i} $s generated using Alg. \ref{alg:madd}, encode
transitions that are never used during sampling. For instance, the ADD
$ t_{(\log_2N)-1} $ as constructed by Alg. \ref{alg:madd}, is only used
by the procedure {\ds} for sampling states $ \tr[N/4] $ (given
$\tr[0]$ and $ \tr[N/2] $) and $ \tr[3N/4] $ (given $ \tr[N/2] $ and $
\tr[N] $). Thus, $ t_{(\log_2N)-1} $ should only be concerned with
states reachable in exactly 0,$ N/4 $, $ N/2 $, $ 3N/4 $ or $ N $
steps from the initial set. However, the $ t_{(\log_2N) - 1} $
constructed by Alg. \ref{alg:madd} also contains information about
other $ 2^{(\log_2N)-1} $-step transitions from states not reachable in
those many step from the initial set. This information is clearly
superfluous and only serves to increase the size of the ADD. Such
 information is present in all $ t_i $s and exists because
the iterative squaring framework of Alg. \ref{alg:madd} squares all
transitions in the loop on lines 2-4 regardless of the initial state,
final state and reachability conditions. We give an improved squaring
framework, presented in Algs. \ref{alg:madd_a} and \ref{alg:crs}. The
idea is to first compute (over-approximations of) sets of states reachable in exactly $ i $
steps from the initial set, for $ 1\le i \le N $
(Alg. \ref{alg:crs}). We then restrict each ADD $ t_i $ by the over-approximations of only those
reachable state sets it depends on (Alg. \ref{alg:madd_a}).

The set $ \rso[i-1] $ for $ 1\le i\le\log_2N $ in Alg. \ref{alg:crs}
represents the set that will be used for restricting the $ X^0 $
variable set of $ t_{i} $, while the set $ \rst[i-1] $ will be used
for restricting the $ X^i $ variable set of $ t_i $. The
(over-approximate) set of states reachable after exactly $ j $ steps
from the initial state, denoted $ r_j $, is computed in line 5
starting from the initial set by taking the (over-approximate) image
under $ t_0 $ of the reachable set after $ j-1 $ steps.  Computing an
exact image is often difficult for large benchmarks, hence an
over-approximation of the image can be used.  The literature contains
a wide spectrum of heuristic techniques that can be used to trade-off
space for time of computation. Once $ r_j $ is computed, we disjoin
the appropriate elements of $ \rso $ and $ \rst $ with $ r $ in lines
7-11. The special case of $ (N/2) ^{th}$ reachable set is handled
separately in lines 12-13.

After $ \rso $ and $ \rst $ sets are computed, we use them to restrict
$ g $ and $ \hat{g} $ in lines 3-4 of Alg. \ref{alg:madd_a}. The
restrict operation is the one proposed in
\cite{coudert1990verifying}. If $ f = Restrict(g,h) $ then $ f = g$
wherever $ h $ is true, and $ f $ is undefined otherwise. This
operation can be more efficient than conjunction, and is sufficient
for our purposes as we explicitly enforce initial state condition in
Line 7 of Alg. \ref{alg:madd_a}.
\begin{algorithm}
	\caption{$makeADDs(t_0, N, \rso, \rst, f, I)$}
	\label{alg:madd_a}
	\begin{algorithmic}[1]
		\INPUT 
		
		$ t_0 $: 1-step transition function
		
		$ N $: trace length 
		
		$ f $: final-state function
		
		$ I $: initial-state function
		
		$ \rso,\rst $: reachable state-sets
		\OUTPUT ADDs $ t_1 \ldots t_{\log_2N}$: $ 2^i $-step transition relations $ t_i $
		\State $ g \gets t_0 $;
		\For {$i = 1, 2, \ldots , \log_2N$}
		\State $ \hat{g}(X^i,X^{i+1}) \gets Restrict(g(X^i,X^{i+1}), \rst[i-1](X^i)) $
		\State $ g(X^{0},X^{i}) \gets Restrict(g(X^{0},X^{i}), \rso[i-1](X^0)) $;
		\State $ t_i (X^{0},X^i,X^{i+1}) \gets g(X^{0},X^{i}) \times \hat{g}(X^i,X^{i+1}) $;
		\State $ g \gets \exists X^i \,\,\, t_i $;
		\Comment Additively abstract vars in $ X^i $
		\EndFor
		\State $ t_{\log_2N} \gets t_{\log_2N} \wedge f(X^{(\log_2N) + 1}) \wedge I(X^0)$
		\State \Return $ t_1 \ldots t_{\log_2N} $
	\end{algorithmic}
\end{algorithm}

\begin{algorithm}
	\caption{$computeReachableSets(t_0, I)$}
	\label{alg:crs}
	\begin{algorithmic}[1]
		\State $ r_0 \gets I $; \Comment Initialize $ r $ to be the initial state set
		\State $ \rso \gets [I,\ldots,I]$;
		
		\Comment Initialize array of $ (\log_2N) $ initial state functions.
		\State $ \rst \gets [0,\ldots,0]$;
		
		\Comment Initialize array of $ (\log_2N) $ Boolean 0 functions. 
		\For {$ j \in \{1,\ldots,N\} $}
		\State $ r_j \gets \textit{\textbf{Im}}(r_{j-1},t_0)$
		\Comment Find (over approx.) image of $ r $ under $ t_0 $
		\For {each $ i \in \{0,1,2,\ldots (\log_2N)-2\}$}
		\If {$j\%(2^{i+1}) == 0$}
		\State $ \rso[i] \gets \rso[i] \lor r_j $
		\Else
		\State $ \rst[i] \gets \rst[i] \lor r_j $
		\State break;
		\EndIf
		\EndFor
		\If {$ j == N/2 $}
		\State $ \rst[(\log_2N) - 1]  = \rst[(\log_2N) - 1] \lor r_j $
		\EndIf
		\EndFor
		\State \Return $\rso,\rst$
	\end{algorithmic}
\end{algorithm}

%% file: Analysis.tex
\section{Analysis}
\label{sec:analysis}

\subsection{Hardness of Counting/Sampling Traces}
Counting and sampling satisfying assignments of an arbitrary Boolean
formula, say $ \phi $, can be easily reduced to counting and sampling,
respectively, of traces of a transition system.  From classical
results on counting and sampling
in~\cite{valiant,stockmeyer83,jerrum1986random,BGP2000}, it follows
that counting traces is \#P-hard and uniformly sampling traces can be
solved in probabilistic polynomial time with access to an NP-oracle.

To see how the reduction works, suppose the support of $\phi$ has $n$
variables, say $x_1, \ldots x_n$.  We construct a transition system
$(S,\Sigma,t, I, F)$, where $S = \{0,1\}^n$ and the set of state
variables is $X = \{x_1, \ldots x_n\}$. We let $\Sigma = \{0,1\}$ and
define the transition function $t:\{0,1\}^n \times
\{0,1\}\rightarrow\{0,1\}^n $as follows: $ t(x_1,\ldots,x_n,a)[0] =
\phi(x_1,\ldots,x_n)$ and $ t(x_1,\ldots,x_n,a)[1] =
t(x_1,\ldots,x_n,a)[2] = \ldots = t(x_1,\ldots,x_n,a)[n] = 0$, for $a
\in \{0,1\}$. In other words, the $0^{th}$ next-state bit is
determined by $\phi$ regardless of the input $a$, while the rest of
the next-state bits are always $0$. We define $I = \{0,1\}^n $ and $F
= \{1000\cdots0\}$.  It is easy to see that counting/sampling traces
of length $1$ of this transition system effectively counts/samples
satisfying assignments of $\phi$.

\subsection{Random Walks and Uniform Traces}
It is natural to ask if uniform trace-sampling can be achieved by a
Markovian random walk, wherein the outgoing transition from a state is
chosen according to a probability distribution specific to the state.
Unfortunately, we show below that this cannot always be done.  Since
uniform sampling is a special case of weighted sampling, the
impossibility result holds for weighted trace sampling too.

Consider the transition system in Fig. \ref{fig:ckt}. We've seen in
Section~\ref{sec:intro} that there are 7 traces of length 4.  Hence a
uniform sampling would generate each of these traces with probability
$1/7$.  Suppose, the probability of transitioning to state $s_j$ from
state $s_i$ is given by $\Pr[(s_i,s_j)]$.  For uniform sampling, we
require $\Pr[(s_i, s_j)] > 0$ if $\exists a \in \Sigma.\, s_j = t(s_i,
a)$, and also $\sum_{s_j\,:\,\exists a,\, s_j = t(s_i,a)}
\Pr[(s_i,s_j)] = 1$. Now, consider the traces $\tr_1 = s_0s_1s_1s_1s_1
$ and $ \tr_2 = s_0s_1s_1s_1s_2 $.  
Let $ \Pr[(s_0,s_1)] = c~(> 0)$ and $ \Pr[(s_1,s_1)] = d~(> 0)$. This
implies that $ \Pr[(s_1,s_2)] = 1-d~(> 0)$. Thus, the probability of
sampling $ \tr_1 $ is $ c.d^3 $. For uniformity, $ c.d^3 =
1/7$. Similarly, from $ \tr_2 $, we get $ c.d^2.(1-d) = 1/7$. From
these two equations, we obtain $ c.d^2 = 2/7 $. Therefore, $d =
\frac{cd^3}{cd^2} = 1/2 $.  It follows from the equation $ cd^3 = 1/7
$ that $ c = 8/7 $.  However, this is not a valid probability measure.
Therefore, it is impossible to uniformly sample traces of this
transition system by performing a Markovian random walk.

\subsection{Correctness of Algorithms}
\begin{table}[]
		\begin{tabular}{l|c|c}
			& \textbf{$X^0$} & \multicolumn{1}{c}{\textbf{$X^i$}}  \\ \hline
			$t_1$ & $j\in \{0,2,4,6,...\}$  & \multicolumn{1}{c}{$j\in \{1,3,5,7,...\}$}   \\ \hline
			$t_2$ & $j\in \{0,4,8,12,...\}$ & \multicolumn{1}{c}{$j\in \{2,6,10,14,...\}$} \\ \hline
			$t_3$ & $j\in \{0,8,16,...\}$   & $j\in \{4,12,20,...\}$ \\ \hline                       
			$...$ & ...   & ...
		\end{tabular}%
	\caption{Reachable sets $ r_{j} $ that $ X^0, X^i $ variables of $ t_i $ depend on}
	\label{tab:reach}
\vspace{-0.6cm}
\end{table}

We now turn to proving the correctness of algorithms presented in the
previous section.  We first prove the correctness of the improved
iterative squaring framework (Sec. \ref{sec:iter}). Alg. \ref{alg:crs}
(lines 8-10) ensures that $ \rso[i-1] $ is computed as a disjunction
of $ r_j $'s for values of $ j $ given in the first column and row $ i
$ of Tab. \ref{tab:reach}, while $ \rst[i-1] $ is computed from $ r_j
$'s for values of $ j $ given on the $ i^{th} $ row and second
column. Therefore, to show the correctness of Algs. \ref{alg:madd_a}
and $ \ref{alg:crs} $, we show in Lemma \ref{lem:crs1} that the $X^0$
and $ X^i $ variable sets of $ t_i $ will only be instantiated with
(over-approximations of) sets of states reachable in the number of
steps given in the appropriate column of Tab. \ref{tab:reach}.

\begin{lemma}
	Let $ S^{p}_{q} $ denote the set of states that the variable set $ X^p $ of $ t_q $ will be instantiated with by Alg. \ref{alg:ds}, when used in conjunction with Algs. \ref{alg:crs} and \ref{alg:madd_a}. Then $ \forall s \in S^{0}_{i} $, we have $ s \in r_j $ for some $ j $ given in column 1 and row $ i $ of Tab. \ref{tab:reach}, and $ \forall s \in S^{i}_{i} $, we have $ s \in r_j $ for some $ j $ given in column 2 and row $ i $ of Tab. \ref{tab:reach}.
	\label{lem:crs1}
\end{lemma}
\begin{proof}
	We show by induction on $ i $ from $ \log_2N $ down to $ 1 $. The base case is shown by the fact that $ t_{\log_2N} $ is used exactly once by {\ds} and the $ X^0 $ variables are used only for sampling the initial state while $ X^i $ is used for sampling $ \tr[N/2] $. Thus $S^0_{\log_2 N} \subseteq r_0 $ and $S^{\log_2 N}_{\log_2 N} \subseteq r_{N/2} $. The former condition is satisfied by the limits of the for-loop in Line 6 of Alg. \ref{alg:crs}, while the latter condition is satisfied by lines 12-13 of Alg. \ref{alg:crs}. This completes the base case. 
	
	Now assume that the lemma holds for some $ i $. We will prove that the lemma holds for $ i-1 $ as well. First note that $ t_i $ is used by {\sfa} for sampling some state $ \tr[m] $ given states $ \tr[m-2^{i-1}] $ and $ \tr[m+2^{i-1}] $. Thereafter, $ t_{i-1} $ is used in 2 cases: (1) for sampling $ \tr[m+2^{i-2}] $ given $ \tr[m] $ and $ \tr[m+2^{i-1}] $; and (2) for sampling $ \tr[m-2^{i-2}] $ given  $ \tr[m] $ and $ \tr[m-2^{i-1}] $. Thus the $ X^0 $ variables of $ t_{i-1} $ will be instantiated with the same states as for $ X^0 $ variables of $ t_i $ in case (1). In case (2), $ X^0 $ vars of $ t_{i-1} $ will be instantiated with the same states as for $ X^i $ vars of $ t_i $. Thus the states instantiating $ X^0 $ vars of $ t_{i-1} $ are the union of the states instantiating $ X_0 $ and $ X^i $ variables of $ t_i $, i.e., $ S^0_{i-1} = S^0_{i} \cup S^i_{i}  $. The values in Tab. \ref{tab:reach} reflect this fact, and by our inductive assumption $S^0_{i}$ and $S^i_{i}$ were computed correctly. This proves that $ \forall s \in S^{0}_{i} $, $ s \in r_j $ for some $ j $ given in column 1 and row $ i $ of Tab. \ref{tab:reach}. To complete the inductive argument we still need to show that $ \forall s \in S^{i}_{i} $, $ s \in r_j $ for some $ j $ given in column 2 and row $ i $ of Tab. \ref{tab:reach}. To see this, first note that the $ X^i $ variables of $ t_i $ will only be instantiated with states reachable in $ 2^{i-1} $ steps from the states instantiating the $ X^0 $ variables of $ t_i $. This is reflected in Tab. \ref{tab:reach}. For instance, in row 3 ($i=3$), $ r_4,r_{12},r_{20}\ldots $ in column 2 are exactly the set of states reachable in $ 2^{i-1} = 4 $ steps from $ r_0,r_{8},r_{16}\ldots $ respectively, in column 1. Since we showed that $ S^{0}_{i} $ has been computed correctly, this completes the proof. 
\end{proof}

Let $ \pc{l}{s_i}{s_j} $ denote the number of traces of length $2^l$
starting in state $ s_i $ and ending in state $ s_j $. Note that $
\pc{l}{s_i}{s_j} = \sum_{s_k\in S} \pc{l-1}{s_i}{s_k} \times
\pc{l-1}{s_k}{s_j} $. We use the fact that {\sfa} ensures that the
parent of a node $ v $ is sampled independently of the path from an
ADD leaf to $ v $ chosen so far. Conditional independence also holds
for whole traces; given the states at two indices in a trace, the
states within the trace segment delineated by the indices are sampled
independently of the states outside the trace segment.  The following
lemmas characterize the behavior of the sampling framework
(Algs. \ref{alg:ds}--\ref{alg:sfa}).
\begin{restatable}{lemma}{lemtwo}
  For $ 1\le i\le \log_2 N $, the ADD $t_i$ computed by $makeADDs$ is
  such that $ \forall s_{j_1},s_{j_2}, s_{j_3} \in S $, we have
  $t_i(s_{j_1},s_{j_2},s_{j_3})$ $=$
  $\pc{i-1}{s_{j_1}}{s_{j_2}}\times \pc{i-1}{s_{j_2}}{s_{j_3}}$
	\label{lem:ti}
\end{restatable}
\begin{proof}
	We will prove by induction on $ i $.
		
		\textit{Base case}: We have $ \forall s_{j_1},s_{j_2} \in S\,\, t_0(s_{j_1},s_{j_2}) = \pc{0}{s_{j_1}}{s_{j_2}} $ by definition. From line 3 of Alg. \ref{alg:madd}, we then have $\forall s_{j_1},s_{j_2}, s_{j_3} \in S$, $t_1(s_{j_1},s_{j_2},s_{j_3}) = \pc{0}{s_{j_1}}{s_{j_2}}\times \pc{0}{s_{j_2}}{s_{j_3}}) $
		
		\textit{Induction step}: Assume the lemma holds up to some $ i $, i.e. $\forall s_{j_1},s_{j_2}, s_{j_3} \in S \,\, t_{i}(s_{j_1},s_{j_2},s_{j_3}) = \pc{i-1}{s_{j_1}}{s_{j_2}}\times \pc{i-1}{s_{j_2}}{s_{j_3}}$. After execution of line 4 of Alg. \ref{alg:madd}, we will have $\forall s_{j_1}, s_{j_3} \in S \,\, g(s_{j_1},s_{j_3}) = \sum_{s_{j_2}} \pc{i-1}{s_{j_1}}{s_{j_2}}\times \pc{i-1}{s_{j_2}}{s_{j_3}} = \pc{i}{s_{j_1}}{s_{j_3}}$. Then in the next iteration of the loop after line 3, we will have  $\forall s_{j_1},s_{j_2}, s_{j_3} \in S \,\, t_{i+1}(s_{j_1},s_{j_2},s_{j_3}) = \pc{i}{s_{j_1}}{s_{j_2}} \times \pc{i}{s_{j_2}}{s_{j_3}} $.  
\end{proof}

\begin{restatable}{lemma}{lemone}
	Let Z denote the random path from a leaf to the root of ADD $
        \hat{t}$ (see Alg. 4) chosen by {\sfa}. Then
	\begin{equation}
	\forall \Path \in \Paths\,\,\,\Pr[Z=\Path] = \frac{val(\Path[0])}{\sum_{v\in leaves(\hat{t})} val(v)\cdot|\Paths_v|}
	\label{eqn:path_prob}
	\end{equation}  
\end{restatable}
\begin{proof}
	The leaf $ v_l $ is sampled with probability $ \Pr[\Path[0]=v_l] =  \frac{val(v_l)\cdot|\Paths_{v_l}|}{\sum_{v\in leaves(\hat{t})} val(v)\cdot|\Paths_v|} $. Thereafter, each parent $ p^* $ is sampled with probability $\Pr[\Path[i]=p^* | \Path[i-1]=v] =  \frac{|\Paths_{p^*}|\cdot2^{\lvlDiff}}{\sum_{p\in P(v)} |\Paths_p|\cdot2^{\lvlDiff}} $, where $ \lvlDiff = level(p)-level(v)-1 $. But note that $ \Paths_{v} = \sum_{p\in P(v)} (|\Paths_p|\cdot2^{\lvlDiff}) $. Then, substituting in the identity $ \Pr[Z=\Path] =  \big(\Pr\big[\Path[0]\big]\cdot\prod_i\Pr\big[\Path[i]\big|\Path[i-1]\big]\big)$, gives the lemma.  
\end{proof}

In the next two lemmas, `$ \lo $' and `$ \hi $' refer to trace indices passed as arguments to {\sfa}, and $ \mi = (\lo + \hi)/2 $.
\begin{restatable}{lemma}{lemthree}
	Suppose {\sfa} is invoked with $ i < \log_2 N$, $\tr[\lo] = s_{j_1}$ and $\tr[\hi] = s_{j_3}$. Let $M$ denote the random state returned by {\sfa} for $\tr[\mi]$. Then
	for all $s_{j_2}\in S$, we have $Pr\big[M = s_{j_2} ~\big|~ \tr[\lo] = s_{j_1},\tr[\hi] = s_{j_3}\big] = \frac{\pc{i-1}{s_{j_1}}{s_{j_2}}\times\pc{i-1}{s_{j_2}}{s_{j_3}}}{\pc{i}{s_{j_1}}{s_{j_3}}}$
	 \label{lem:illn}
\end{restatable}
\begin{proof}
	We note that for any ADD $ t_i $ s.t. $ i < \log_2N $, we
        reduce the ADD by substituting $ \tr[\lo],\tr[\hi] $ in line 4
        of Alg. \ref{alg:sfa}. In the resultant ADD $ \hat{t} $, each
        paths from root to leaf yields a valuation for $ \tr[\mi]
        $. Therefore, if $ \Path $ is the path traversed in $ \hat{t}
        $ corresponding to some state $ s_{j_2} $, then $
        Pr\big[\tr[\mi] = s_{j_2}\big|\tr[\lo] = s_{j_1},\tr[\hi] =
          s_{j_3}\big] = \Pr[ Z = \Path ] $. We now need to prove that
        the R.H.S. of Eqn. \ref{eqn:path_prob} is the same as the
        desired conditional probability expression. In
        Eqn. \ref{eqn:path_prob}, the numerator $ val(\Path[0]) =
        t_i(s_{j_1},s_{j_2},s_{j_3}) =
        \pc{i-1}{s_{j_1}}{s_{j_2}}\times \pc{i-1}{s_{j_2}}{s_{j_3}}$,
        by Lemma \ref{lem:ti}. The denominator of
        Eqn. \ref{eqn:path_prob} is $ \sum_{v\in leaves(\hat{t})}
        (val(v)*|\Paths_v|) =\sum_{s_{j_2}}
        \pc{i-1}{s_{j_1}}{s_{j_2}}\times \pc{i-1}{s_{j_2}}{s_{j_3}}$
        which is $ \pc{i}{s_{j_1}}{s_{j_3}} $.
\end{proof}
\begin{restatable}{lemma}{lemfour}
  Let {\sfa} be invoked with $ i = \log_2 N $, and let $L$, $M$ and $H$
  denote the random states returned for $\tr[\lo]$, $\tr[\mi]$ and $\tr[\hi]$
  respectively. Then 
	for all $s_{j_1},s_{j_2},s_{j_3}\in S$ s.t. $I(s_{j_1})$ and  $f(s_{j_3})$ hold, we have $\Pr\big[(L =s_{j_1})\wedge(M =s_{j_2})\wedge(H =s_{j_3})\big]
	=\frac{\pc{i-1}{s_{j_1}}{s_{j_2}}\times\pc{i-1}{s_{j_2}}{s_{j_3}}}{|\Tr_N|}$
	\label{lem:ieln}
\end{restatable}
\begin{proof}
	By definition, $ |\Tr_N| = \sum_{s_{j_1},s_{j_3}}
	\pc{N}{s_{j_1}}{s_{j_3}} $, when $ s_{j_1} \models I $ and $ s_{j_3}
	\models f $. The rest of the proof is similar to that of Lem.
	\ref{lem:illn}.
\end{proof}
\begin{restatable}{theorem}{mainthm}
  \label{thm:main}
	Let $ Y $ be a random trace returned by an invocation of {\ds}. For all $ \tr \in \Tr_N $, we have $ \Pr[Y=\tr] = \frac{1}{|\Tr_N|} $.
\end{restatable}
 \begin{proof}
 	Recursively halving $ \tr $, we get $ \Pr[Y = \tr]  = \Pr[(\tr[0]=s_{j_1})\wedge(\tr[N/2]=s_{j_2})\wedge(\tr[N]=s_{j_3})]\cdot\Pr[(\tr[N/4]=s_{j_4}) | (\tr[0]=s_{j_1}) \wedge (\tr[N/2]=s_{j_2})]\cdot\Pr[(\tr[3N/4]=s_{j_5}) | (\tr[N/2]=s_{j_2}) \wedge (\tr[N]=s_{j_3})]\ldots$ Substituting values in the RHS from Lemmas \ref{lem:illn} and \ref{lem:ieln}, we get the result by noting that $\forall s_{j_1},s_{j_2}\in S\,\, \pc{0}{s_{j_1}}{s_{j_2}} \in \{0,1\}$ since the transition system is deterministic.
 \end{proof}

%% file: Experiments.tex
\section{Empirical Evaluation}
We have implemented our algorithms in a tool called {\ts}. The objective of our empirical study was to compare {\ts}\footnote{Code available at \url{https://gitlab.com/Shrotri/tracesampler}} with other state-of-the-art approaches in terms of number of benchmarks solved as well as speed of solving.

\paragraph{\bfseries Experimental Setup}
As noted in Section~\ref{sec:rel}, {\uniw}~\cite{chakraborty2013scalable}, {\uni} \cite{chakraborty2014balancing} and
{\utwo}~\cite{chakraborty2015parallel} are state-of-the-art tools
for almost uniform sampling and
SPUR~\cite{achlioptas2018fast}, KUS~\cite{sharma2018knowledge} and
WAPS~\cite{gupta2019waps} are similar tools for uniform sampling of
SAT witnesses.
We compare {\ts} with {\utwo} and WAPS in our experiments, since these
are currently among the best almost-uniform and uniform samplers
respectively, of SAT witnesses. We invoke both WAPS and {\utwo} with
default settings. Although {\utwo} is capable of operating in
parallel, we invoke it in serial mode to ensure fairness of comparison.

We ran all our experiments on a high performance cluster. Each experiment
had access to one core on an Intel Xeon E5-2650 v2 processor running
at 2.6 GHz, with 4GB RAM. We used GCC 6.4.0 for compiling {\ts} with
O3 flag enabled, along with CUDD library version 3.0 with dynamic
variable ordering enabled. We set a timeout of 7200 seconds for each
experiment. For experiments that involved converting benchmarks in
Aiger format to BDD (explained below), we allotted 3600 seconds
out of 7200 exclusively for this conversion. We attempted to generate
5000 samples in each instance. We called an experiment successful or
completed, if the sampler successfully sampled 5000 traces within the
given timeout.

\paragraph{\bfseries Benchmarks}
We used sequential circuit benchmarks from the Hardware Model Checking
Competition~\cite{biere2010hardware} and
ISCAS89~\cite{brglez1989notes} suites. Each benchmark represents a
sequential circuit in the And-Inverter Graph (AIG) format.  In
general, primary outputs of such a circuit can indicate if target
states have been reached, and can be used to filter the set of traces
from which we must sample. In our experiments, however,
we ignored the primary outputs, and sampled from all traces of a given
length $N$ starting from the all-zero starting
state. We attempt uniform sampling of traces in our
  experiments, as the benchmarks do not provide weights for
  transitions.

As mentioned in Section~\ref{sec:prelim}, we need to existentially
quantify the primary inputs from the transition functions to get the
transition relations. This is done either explicitly or implicitly
depending on the sampler to be used. {\ts} requires the transition
relation $ t $ in the form of a BDD while WAPS and {\utwo} require a
CNF formula. We used a straightforward recursive procedure for
converting $ t $ (as AIG) to a BDD. We then quantified out the primary
inputs using library functions in CUDD. For converting $ t $ to CNF
there were two choices: (1) by obtaining the prime cover using a
built-in operation in CUDD, or (2) using Tseitin encoding to convert
the AIG to CNF by introducing auxiliary variables. The CNF obtained
from the first method has no auxiliary variables that need to be
existentially quantified; hence it can be used with WAPS and the D4
compiler, which does not support existential
quantification.\footnote{WAPS also can work with the DSharp compiler,
  which supports existential quantification. However, in our
  experiments we found that DSharp provided incorrect answers.} In
contrast, the second method can be used in conjunction with {\utwo},
since the auxiliary variables need to be projected out.

\version{
\begin{figure}
	\centering
	\includegraphics[width=2.5in]{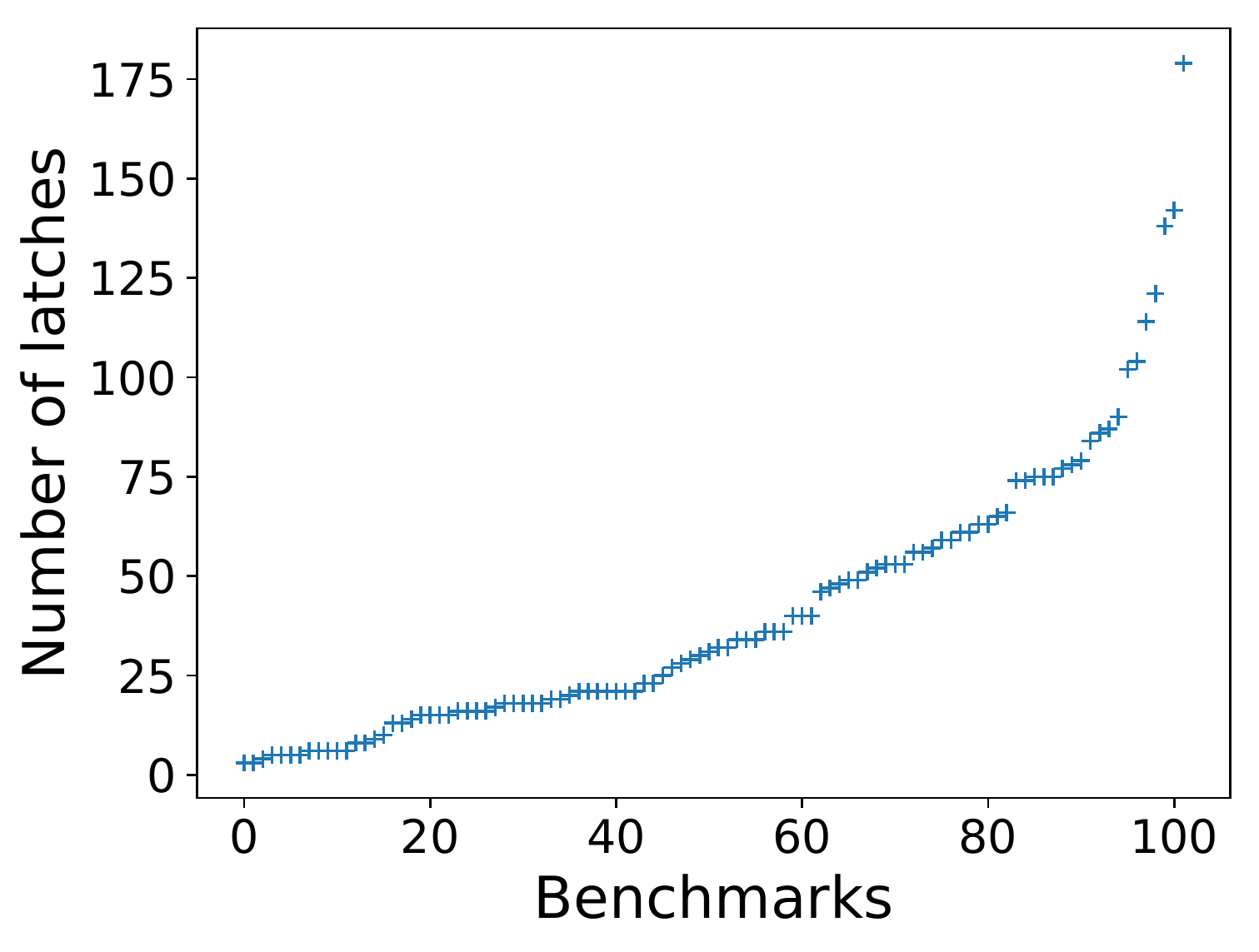}
	\caption{Distribution of benchmark sizes (number of latches)}
	\label{fig:bench_stats}
\end{figure}}{}
Within the available time and memory, we obtained a total
of 310 pre-processed AIG files, out of which 102 could be converted
to BDDs with primary inputs existentially quantified out. The number of latches for these 102 benchmarks ranged between 5 and 175, and the median number of latches was 32. \version{The distribution of number of of latches is depicted in Fig. \ref{fig:bench_stats}.}{} We
restricted the start state to be all-zeros since this is implicit in
the AIG format. For each benchmark, we attempted to sample traces
with lengths 2, 4, 8, 16, 32, 64, 128 and 256. We chose this range of
trace lengths since a vast majority of benchmarks in HWMCC-17
(particularly, benchmarks in the DEEP category) required bounds within
256~\cite{hwmcc-bounds}. Further, we observed that none of the tools
were able to consistently scale beyond traces of length 256. We refer
to a benchmark and a given trace length as an 'instance'. We thus
generated $102\times8 = 816$ instances for BDD based approaches. For
CNF based approaches, the unrolling was done by appropriately
unrolling the transition relation. Note that the first CNF-based
approach was applicable to only 102 benchmarks that could be converted
to BDDs, while the direct conversion from AIG to CNF was technically
possible for all 310 benchmarks. However, we primarily report on the
816 instances even for AIG-encoded CNFs. We discarded formulas with
more than $10^6$ clauses, as the files became too large.

\paragraph{\bfseries Results}
In our experiments we found that {\utwo} fared better with formulas
encoded from BDDs, vis-a-vis formulas encoded directly from
AIGs. All reported results are, therefore, on
BDD-encoded formulas. We only report on instances
with at least $100$ distinct traces of the given length, since trace sampling can be trivially implemented by enumerating traces, if
the trace-count is small.

We consistently found that {\ts} outperformed both WAPS and {\utwo} by
a substantial margin. We present a comparison of the performance of
the 3 tools on the 816 instances where BDD construction succeeded.  Figure \ref{fig:completed} shows
a cactus plot of the number of experiments completed in the given
time, with the number of instances on x--axis and the total time
taken on y--axis. A point $(x, y)$ implies that $x$ instances took
less than or equal to $y$ seconds to solve. {\ts} is able
to complete 502 experiments out of 816 -- almost 200 more than WAPS and 350 more than {\utwo}.

\version{
\begin{figure}
	\centering
	\includegraphics[width=2.5in]{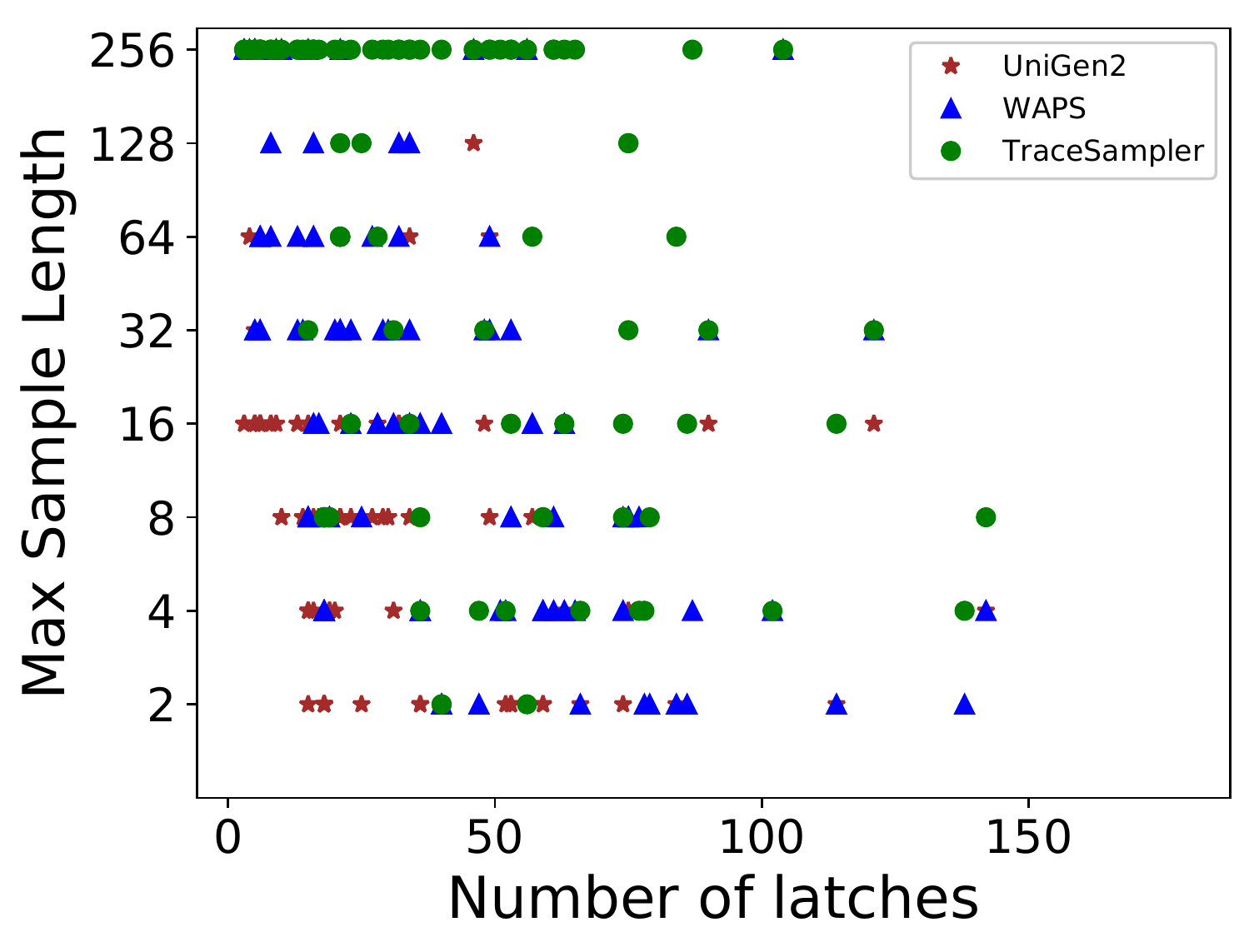}
	\caption{Length of longest trace sampled vs. number of latches for each benchmark}
	\label{fig:lenvlatch}
\end{figure}}{}
{\ts} was also fastest on the majority of instances.  Among a total of
503 instances on which at least one sampler succeeded, {\ts} was
fastest on 446 (88.7\%) while WAPS and {\utwo} were fastest on 33
(6.5\%) and 24 (4.8\%) respectively. For instances on which both tools
were successful, we found that on average (geometric mean) {\ts}
offered a speedup of 25$\times$ compared to {\utwo} and 3 relative
to WAPS. Overall, {\ts} was able to sample traces 3.5 times longer on
average (geometric mean) as compared to WAPS and 10 times longer as
compared to {\utwo}. Further, {\ts} is able to sample traces of length 256 from 52 benchmarks, while WAPS and {\utwo} are able to sample 256-length traces from 12 and 3 benchmarks respectively.\version{ Fig. \ref{fig:lenvlatch} depicts the distribution of the maximum length of traces each algorithm is able to sample from, relative to the size (number of latches) of the corresponding benchmarks. It can be seen that {\ts} is generally able to sample longer traces from larger benchmarks.}{}

WAPS and {\ts} proceed in two phases --- the compilation phase where a d-DNNF or ADDs are constructed, and the sampling phase where the constructed structures are traversed. When only considering the time required for compilation, {\ts} offered speedup factor of 16 compared to WAPS.

\begin{figure}
	\centering
	\includegraphics[width=2.5in]{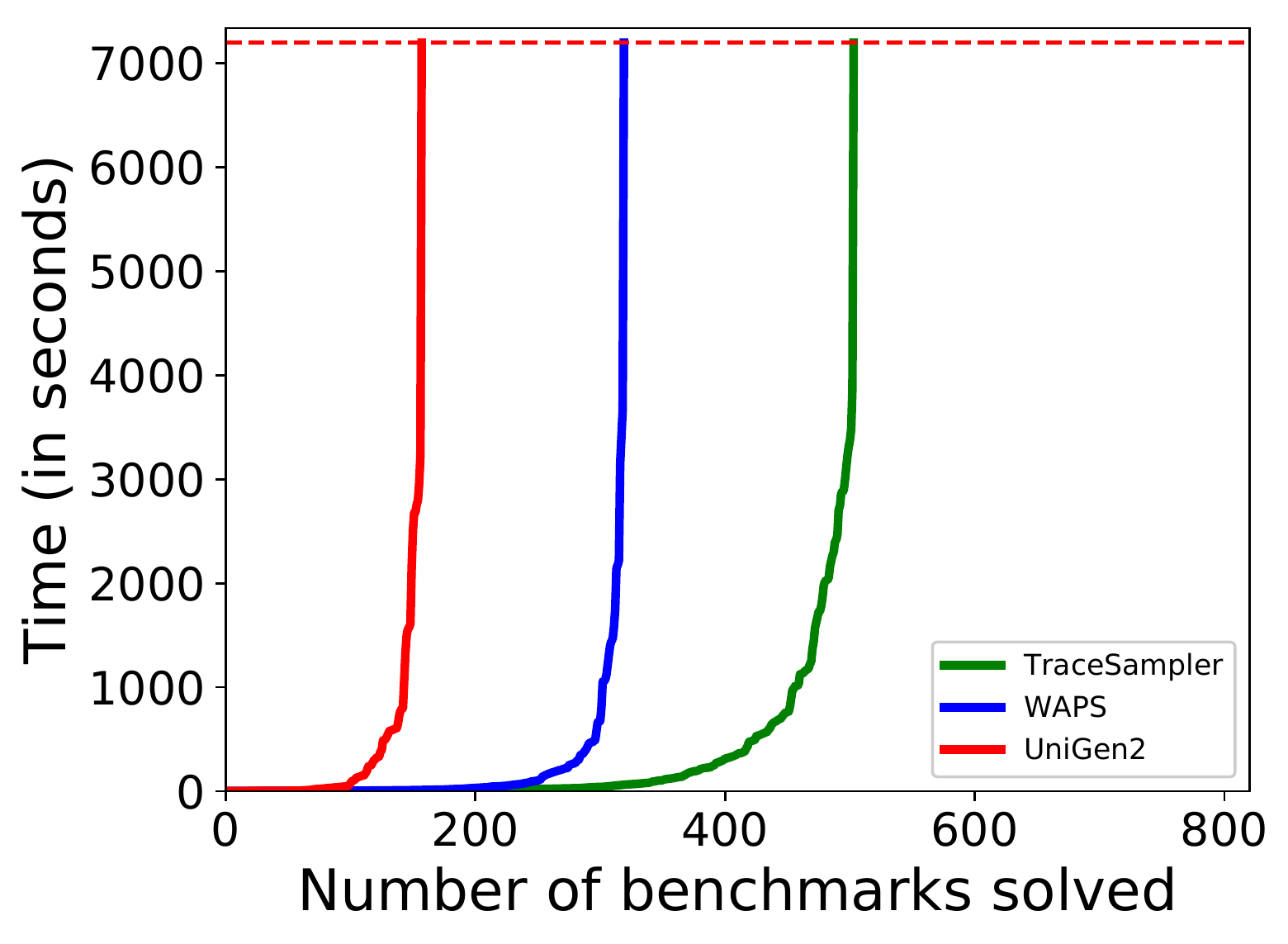}
	\caption{Performance Comparison of {\ts} with WAPS and {\utwo}.}
	\label{fig:completed}
\end{figure}
\version{
\paragraph{\bfseries Output Distribution}
While Theorem~\ref{thm:main} guarantees uniformity of distribution of
traces generated by {\ts}, we performed a simple experiment to compare
the actual distribution of traces generated by {\ts} with that
generated by WAPS -- a perfectly uniform sampler. The instance we
selected had 8192 distinct traces. We generated $ 10^6 $ traces samples using
both {\ts} and WAPS, 
computed the frequency of occurrence of each trace, and grouped traces
occurring with the same frequency. This is shown in
Fig. \ref{fig:dist} where a point $(x,y)$ indicates that $ x $
distinct traces were generated $ y $ times. It can be seen that the
distributions generated by {\ts} and WAPS are practically
indistinguishable, with Jensen-Shannon distance 0.003. Similar trends
were observed for other benchmarks as well.
\begin{figure}
	\centering
	\includegraphics[width=2.3in]{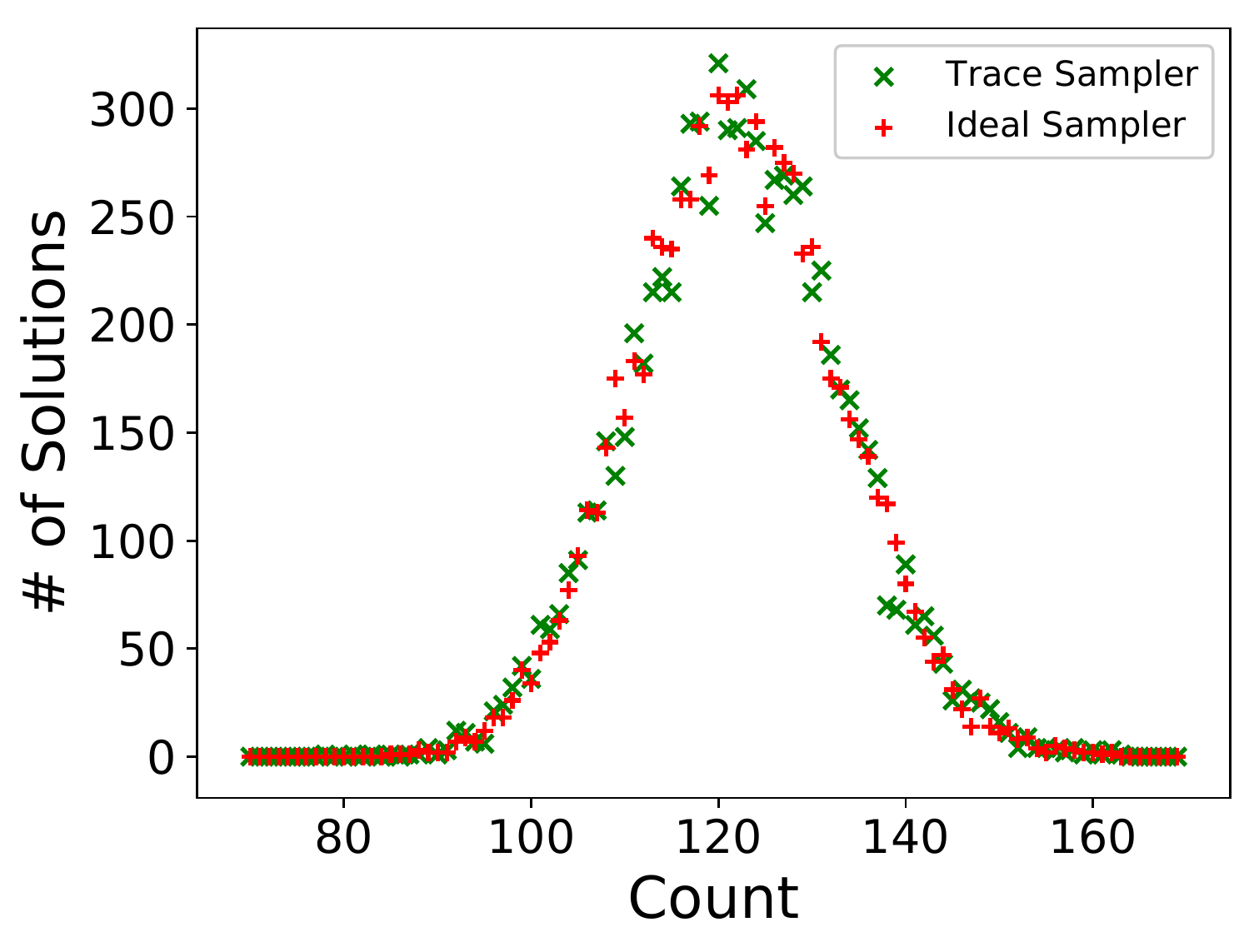}
	\caption{Distributions of generated samples.}
	\label{fig:dist}
\end{figure}
}{}

\paragraph{\bfseries Comparison with ApproxMC3}
The {\uni} series of samplers are based on the approximate counting
tool {\amc}. 
At the time of writing this paper, the latest version of {\amc},
called {\athree}, had not been incorporated into {\uni}. In order to
obtain an idea of the kind of performance gains one can expect from
{\uni} with updated counting sub-modules, we ran experiments to count
the number of traces of a given length with {\athree}. {\ts} was able
to sample traces $ 5\times $ longer than what {\athree} could count, while providing $17\times $ speedup, on average. Thus,
on benchmarks where BDD construction was successful, {\ts} was clearly
the best choice. However, {\athree} was able to count the number of
$8$-long traces in 62 cases out of 208 in which BDD construction
failed. Yet, the time required for sampling using {\utwo} usually far
exceeds that required to count, as the counting subroutine is invoked
multiple times for obtaining the desired number of samples. The times
reported for {\athree} therefore, are a generous lower-bound on the
times that would be required for sampling.

\paragraph{\bfseries Discussion}
Our experiments show that {\ts} is the algorithm of choice for
uniformly sampling traces and is even able to outperform the state-of-the-art model counter. We consistently observed that most of the
time used for ADD construction is spent in dynamic variable
reordering; given a good variable order, ADD construction is usually
very fast. In industrial settings, a good variable order may
be available for the circuits of interest. In addition, compilation
can often be done 'off-line' resulting in its cost getting amortized
over the generated samples. In this light, the compilation time
speedup of {\ts} relative to WAPS is encouraging.
 
A drawback of using BDDs and ADDs is that they often blow-up in
size. Indeed, we found that conversion of AIG to BDD
failed on 208 benchmarks. Nevertheless, we found that {\utwo} was also
unable to finish sampling a single instance (with trace length 8) of these 208 benchmarks, as well. This indicates that the problem may lie deeper in the  transition structure, rather than in the variable order.

It is worth noting that CRV runs typically span hundreds of thousands
of clock cycles, while we (and other approaches that provide
uniformity guarantees) can sample traces of a few hundred transitions
at present.  This is because trace sampling requires solving global
constraints over the entire length of the trace, while in CRV, local
constraints over short segments of an otherwise long trace need to be
solved.  We believe these are complementary strengths that can be used
synergistically.  Specifically, a CRV tool can be used to drive a
system into a targeted (possibly bug-prone) corner over a large number
of clock cycles.  Subsequently, one can ensure provably good coverage
of the system's runs in this corner by uniformly sampling traces for
the next few hundred cycles.  We believe this synergy can be very
effective in simulation-based functional verification.

%% file: Conclusion.tex
\section{Conclusion}
In this paper, we introduced a symbolic algorithm based on ADDs for
sampling traces of a transition system (sequential circuit) with
provable uniformity (or bias) guarantees. We demonstrated its
scalability vis-a-vis other competing approaches that provide similar
guarantees, through an extensive empirical study.  Our experience
indicates that there is significant potential to improve the
performance of our tool through engineering optimizations. 
An interesting direction for further research is to combine the strengths of
decision diagram based techniques (like {\ts}) with SAT-solving based
techniques (like {\utwo}) to build trace samplers that have the best
of both worlds.